\documentclass[smallextended]{svjour3}%
\pdfoutput=1
\usepackage{amsfonts}
\usepackage{amsmath}
\usepackage{amssymb}
\usepackage{graphicx}
\usepackage{fullpage}%
\providecommand{\U}[1]{\protect\rule{.1in}{.1in}}
\pdfoutput=1
\smartqed
\journalname{Foundations of Physics}
\begin{document}

\title{Perfect state distinguishability and computational speedups with postselected
closed timelike curves}
\author{Todd A. Brun
\and Mark M. Wilde}
\institute{Todd Brun is a professor with the Center for Quantum Information Science and Technology and the
Communication Sciences Institute of the Ming Hsieh Department of Electrical
Engineering at the University of Southern California, Los Angeles, California
90089 USA. Mark M. Wilde is
a postdoctoral fellow in the School of Computer Science, McGill University, Montreal, Quebec, Canada (E-mail: tbrun@usc.edu; mark.wilde@mcgill.ca).}
\date{Received: \today / Accepted: }
\maketitle

\begin{abstract}
Bennett and Schumacher's postselected quantum teleportation is a model of
closed timelike curves (CTCs) that leads to results physically different from
Deutsch's model. We show that even a single qubit passing through a
postselected CTC (P-CTC) is sufficient to do any postselected quantum
measurement with certainty, and we discuss an important difference between \textquotedblleft
Deutschian\textquotedblright\ CTCs (D-CTCs) and P-CTCs in which the future
existence of a P-CTC might affect the present outcome of an experiment. Then,
based on a suggestion of Bennett and Smith, we explicitly show how a party
assisted by P-CTCs can distinguish a set of linearly independent quantum
states, and we prove that it is not possible for such a party to distinguish a
set of linearly dependent states. The power of P-CTCs is thus weaker than that
of D-CTCs\ because the Holevo bound still applies to circuits using them,
regardless of their ability to conspire in violating the uncertainty
principle. We then discuss how different notions of a quantum mixture that are
indistinguishable in linear quantum mechanics lead to dramatically differing
conclusions in a nonlinear quantum mechanics involving P-CTCs. Finally, we
give explicit circuit constructions that can efficiently factor integers,
efficiently solve any decision problem in the intersection of NP and coNP, and
probabilistically solve any decision problem in NP. These circuits accomplish
these tasks with just one qubit traveling back in time, and they exploit the
ability of postselected closed timelike curves to create grandfather paradoxes
for invalid answers. 

\end{abstract}

\section{Introduction}

Einstein's field equations for general relativity predict the existence of
closed timelike curves (CTCs) in certain exotic spacetime geometries
\cite{RevModPhys.21.447,B80,PhysRevLett.66.1126}, but the bizarre consequences
lead many physicists to doubt that such \textquotedblleft time
machines\textquotedblright\ could exist. Closed timelike curves, if they
existed, would allow particles to interact with their former selves,
suggesting the possibility of grandfather-like paradoxes in both classical and
quantum theories. Physicists have considered the ramifications of closed
timelike curves for quantum mechanics by employing path-integral approaches in
an effort to avoid contradictions~\cite{PhysRevLett.61.1446,PhysRevD.49.6543}.

Deutsch showed that closed timelike curves also have consequences for
classical and quantum computation~\cite{PhysRevD.44.3197}, and he suggested
imposing a self-consistency condition on the density matrix\ of a CTC\ qubit
in order to avoid grandfather-like paradoxes. Since Deutsch's seminal work,
quantum information theorists have produced a flurry of results under his
model. They have shown that \textquotedblleft Deutschian\textquotedblright%
\ closed timelike curves (D-CTCs) can help solve NP-complete problems
\cite{PhysRevA.70.032309}, that a D-CTC-assisted classical or quantum computer
has computational power equivalent to that of PSPACE
\cite{ScottAaronson02082009}, that a D-CTC-assisted quantum computer can
perfectly distinguish an arbitrary set of non-orthogonal states
\cite{PhysRevLett.102.210402}, that evolutions of chronology-respecting qubits
can be a discontinuous function of the initial state \cite{DFI10}, and that it
is not possible to purify mixed states of qubits that traverse a D-CTC while
still being consistent with interactions with chronology-respecting
qubits~\cite{PCA10}. The result of Brun \textit{et al}%
.~\cite{PhysRevLett.102.210402} concerning state distinguishability is perhaps
the most striking for any firm believers in unitarity, considering that a
D-CTC-assisted quantum computer can violate both the uncertainty principle and
the Holevo bound~\cite{holevo}.

Since these findings, Bennett \textit{et al}.~\cite{PhysRevLett.103.170502}
questioned the above results of Aaronson and Watrous and Brun \textit{et
al}.~on D-CTC-assisted computation and distinguishability, respectively. They showed that
the circuits of Aaronson \textit{et al}.~do not operate as
advertised when acting on a classically-labeled mixture of
states and argued that this implies their circuits \textquotedblleft become
impotent\textquotedblright\ \cite{BLSS10}. In their work, they exploited \textit{linear} mixtures
of states to suggest that the aforementioned authors fell into a
\textquotedblleft linearity trap.\textquotedblright\ But recent papers cast
doubt on the claims of Bennett \textit{et al}.~and come to the same conclusion
as Aaronson and Watrous and Brun \textit{et al}.~\cite{RM10,CM10}---a first
paper tracks the information flow of quantum systems in a D-CTC\ with a
Heisenberg-picture approach \cite{RM10}, and another paper shows how a density
matrix description is not valid in a nonlinear theory \cite{CM10}. Further
work revisits Deutsch's self-consistency conditions \cite{WB10}, showing that
they are concealing paradoxes from an observer rather than eliminating them as
they should. These dramatically differing conclusions have to do with the
ontological status of quantum states, which, for the most part, is not a major
concern in standard linear quantum mechanics, but clearly leads to differing
results in a nonlinear quantum mechanics.

Recently, a different model of closed timelike curves has
emerged~\cite{S09,Lloyd2010,LMGGS10}, based on Bennett and Schumacher's
well-known but unpublished work on postselected quantum
teleportation~\cite{B05}. This alternative theory features a postselected
closed timelike curve (P-CTC), which is physically inequivalent to a D-CTC
\cite{Lloyd2010,LMGGS10}. Sending a qubit into the past by a P-CTC is somewhat like
teleporting the qubit's state \cite{BBCJPW93}. Normally, states can only be
teleported forward in time, because the receiver requires a measurement
outcome from the sender in order to recover the state. By somehow postselecting with certainty on
only a single measurement outcome, however, this requirement is removed.
Postselection of quantum teleportation in this fashion implies that an
entangled state effectively creates a noiseless quantum channel into the past.
P-CTCs have the benefit of preserving correlations with external systems,
while also being consistent with path-integral formulations of
CTCs~\cite{Lloyd2010,LMGGS10,PhysRevD.49.6543}. Lloyd \textit{et al}.~have
proven that the computational power of P-CTCs is equivalent to that of the
complexity class PP~\cite{Lloyd2010,LMGGS10}, by invoking Aaronson's results
concerning the power of quantum computation with postselection~\cite{A05}. In
this paper, we show that the same result can be derived from a different
direction: by invoking the ideas of Ref.~\cite{fpl2003brun} to eliminate
invalid answers to a decision problem by making them \textquotedblleft
paradoxical.\textquotedblright\ One can exploit this particular aspect of
P-CTCs to give explicit constructions of P-CTC-assisted circuits with dramatic
computational speedups.

Our first result is to show that one can postselect with certainty the outcome of any
generalized measurement using just one P-CTC\ qubit. Lloyd \textit{et
al}.~state that it is possible to perform any desired postselected quantum
computation with certainty with the help of a P-CTC\ system~\cite{Lloyd2010,LMGGS10}, but
they did not explicitly state that it requires just one P-CTC\ qubit. Next, we
discuss a difference between D-CTCs and P-CTCs, in which the existence of a
future P-CTC\ might affect the outcome of a present experiment. This
observation might potentially lead to a way that one could test for a future
P-CTC, by noticing deviations from expected probabilities in quantum
mechanical experiments.

Further results concern state distinguishability with P-CTC-assisted circuits.
We begin by showing that the SWAP-and-controlled-Hadamard circuit from
Ref.~\cite{PhysRevLett.102.210402}\ can perfectly distinguish $\left\vert
1\right\rangle $ and $\left\vert +\right\rangle $ when assisted by a P-CTC
(recall that this circuit can distinguish $\left\vert 0\right\rangle $ and
$\left\vert -\right\rangle $ when assisted by a D-CTC
\cite{PhysRevLett.102.210402}). We show that the circuit from
Ref.~\cite{PhysRevLett.102.210402}\ for distinguishing the BB84 states
$\left\vert 0\right\rangle $, $\left\vert 1\right\rangle $, $\left\vert
+\right\rangle $, and $\left\vert -\right\rangle $ when assisted by a D-CTC
cannot do so when assisted by a P-CTC. The proof of
Theorem~\ref{thm:P-CTC-state-distinguish}\ then constructs a P-CTC-assisted
circuit, similar to the general construction from
Ref.~\cite{PhysRevLett.102.210402}, that can perfectly distinguish an
arbitrary set of linearly independent states. The proof offers an alternate
construction that accomplishes this task with just one P-CTC\ qubit, by
exploiting the generalized measurement of Ref.~\cite{C98}\ and the ability of
a P-CTC\ to postselect with certainty on particular measurement outcomes. The theorem also
states that no P-CTC-assisted circuit can perfectly distinguish a set of
linearly dependent states. Bennett and Smith both suggested in private
communication~\cite{B09,S10} that such a theorem should hold true. The theorem
implies that a P-CTC-assisted circuit cannot beat the Holevo bound, so that
their power is much weaker than that of a D-CTC-assisted circuit for this task
\cite{PhysRevLett.102.210402}. We then discuss how different representations
of a quantum state in P-CTC-assisted circuits lead to dramatically differing
conclusions, even though they give the same results in linear quantum mechanics.

Our final set of results concerns the use of P-CTC-assisted circuits in
certain computational tasks. We first show that a P-CTC-assisted circuit can
efficiently factor integers without the use of the quantum Fourier transform.
We then generalize this result to a P-CTC-assisted circuit that can
efficiently solve any decision problem in the intersection of NP and co-NP.
Our final construction is a P-CTC-assisted circuit for probabilistically
solving any problem in NP. All of our circuits can accomplish these
computational tasks using just one P-CTC\ qubit. These circuits exploit the
idea in Ref.~\cite{fpl2003brun} of making invalid answers paradoxical, which
yields results that are surprisingly similar to Aaronson's construction in
Ref.~\cite{A05} concerning the power of postselected quantum computation.

We structure this paper as follows. The next section briefly reviews the P-CTC
model, and we prove that a single qubit in a P-CTC allows postselecting on any
measurement outcome with certainty. We also discuss an important difference between D-CTCs and P-CTCs
and provide an example to illustrate this difference.
Section~\ref{sec:P-CTC-distinguish}\ presents our results for P-CTCs and state
distinguishability, and Section~\ref{sec:P-CTC-compute}\ presents our results
for P-CTCs in certain computational tasks. We end by summarizing our results.

\section{The Theory of Postselected Closed Timelike Curves}

\label{sec:review}We first briefly review the theory of
P-CTCs~\cite{B05,S09,Lloyd2010,LMGGS10}. A P-CTC-assisted circuit operates by
combining a chronology respecting qubit in a state $\rho$ with a chronology-violating qubit and interacting
them with a unitary evolution. After the unitary, the chronology-respecting qubit proceeds forward in time
while the chronology-violating qubit goes back in time. The assumption of the model is that this evolution
is mathematically equivalent to combining the state $\rho$ with
a maximally
entangled Bell state $\left\vert \Phi\right\rangle $ where%
\[
\left\vert \Phi\right\rangle \equiv\frac{1}{\sqrt{d}}\sum_{i=0}^{d-1}%
\left\vert i\right\rangle \left\vert i\right\rangle .
\]
There is then a unitary interaction $U$\ between the CR qubit and half of the
entangled state. The final step is to project the two systems of the entangled
state onto the state $\left\vert \Phi\right\rangle $, renormalize the state,
and trace out the last two systems. The renormalization induces a nonlinearity
in the evolution.
This approach is the same as the controversial
\textquotedblleft final state projection\textquotedblright\ method from the
theory of black hole evaporation~\cite{HM04,PhysRevLett.96.061302}.
Figure~\ref{fig:PCTC-operation}\ depicts the operation of a P-CTC.

\begin{figure}[ptb]
\begin{center}
\includegraphics[
natheight=1.326600in,
natwidth=2.700000in,
height=0.7671in,
width=1.5342in
]{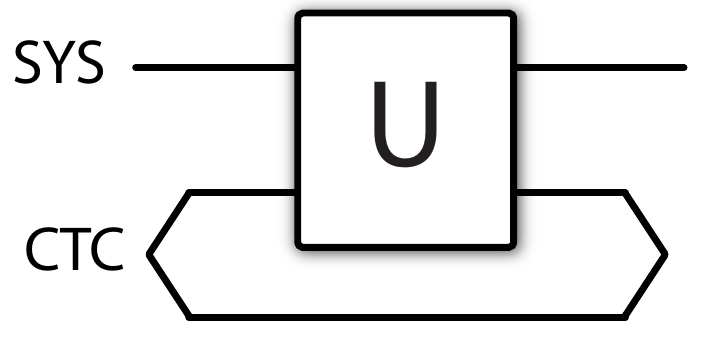}
\end{center}
\caption{The interpretation of the above circuit
is that a chronology-respecting qubit named SYS and a chronology-violating qubit CTC interact
according to a unitary evolution. After the unitary interaction,
the chronology-respecting qubit proceeds forward in time,
while the chronology-violating qubit goes back in time, forming a closed loop.
The assumption of the postselected closed timelike curve (P-CTC) model is that
the above sequence of events is mathematically
equivalent to a ``postselected teleportation protocol.'' 
In such a protocol, the
circuit begins with a chronology-respecting qubit named SYS\ and a maximally
entangled state $\left\vert \Phi\right\rangle $. We then feed the SYS\ system
and half of the maximally entangled state into a unitary $U$. The final step
is to project the two systems of the maximally entangled state onto the state
$\left\vert \Phi\right\rangle $ and renormalize. The assumption of the P-CTC model
is that this final projection is certain and not probabilistic and
that this procedure is mathematically equivalent to
creating a \textquotedblleft quantum channel into the past\textquotedblright%
\ which the CTC\ qubit can exploit to go back in time.}%
\label{fig:PCTC-operation}%
\end{figure}

As pointed out by Lloyd \textit{et al}.~\cite{LMGGS10},\ the action of a
unitary $U_{\text{SYS},A}$\ on a joint system consisting of a
chronology-respecting pure state $\left\vert \psi\right\rangle _{\text{SYS}}$
and a CTC\ system$~A$ is as follows (before renormalization):%
\begin{align*}
\left\langle \Phi\right\vert _{AB}U_{\text{SYS},A}\left(  \left\vert
\psi\right\rangle _{\text{SYS}}\otimes\left\vert \Phi\right\rangle
_{AB}\right)    & =\frac{1}{d}\sum_{i,j}\left\langle i\right\vert
_{A}\left\langle i\right\vert _{B}U_{\text{SYS},A}\left\vert \psi\right\rangle
_{\text{SYS}}\left\vert j\right\rangle _{A}\left\vert j\right\rangle _{B}\\
& =\frac{1}{d}\sum_{i,j}\left\langle i\right\vert _{A}U_{\text{SYS}%
,A}\left\vert j\right\rangle _{A}\left\vert \psi\right\rangle _{\text{SYS}%
}\langle i\left\vert j\right\rangle _{B}\\
& =\frac{1}{d}\sum_{i}\left\langle i\right\vert _{A}U_{\text{SYS},A}\left\vert
i\right\rangle _{A}\left\vert \psi\right\rangle _{\text{SYS}}\\
& =\frac{1}{d}\text{Tr}_{A}\left\{  U_{\text{SYS},A}\right\}  \left\vert
\psi\right\rangle _{\text{SYS}},
\end{align*}
implying the following evolution for a mixed state $\rho_{\text{SYS}}$ (before
renormalization):%
\begin{align*}
&  \left\langle \Phi\right\vert _{AB}U_{\text{SYS},A}\left(  \rho_{\text{SYS}%
}\otimes\left\vert \Phi\right\rangle \left\langle \Phi\right\vert
_{AB}\right)  U_{\text{SYS},A}^{\dag}\left\vert \Phi\right\rangle _{AB}\\
&  =\left\langle \Phi\right\vert _{AB}U_{\text{SYS},A}\left(  \sum_{i}%
\lambda_{i}\left\vert \psi_{i}\right\rangle \left\langle \psi_{i}\right\vert
_{\text{SYS}}\otimes\left\vert \Phi\right\rangle \left\langle \Phi\right\vert
_{AB}\right)  U_{\text{SYS},A}^{\dag}\left\vert \Phi\right\rangle _{AB}\\
&  =\sum_{i}\lambda_{i}\left\langle \Phi\right\vert _{AB}U_{\text{SYS}%
,A}\left(  \left\vert \psi_{i}\right\rangle \left\langle \psi_{i}\right\vert
_{\text{SYS}}\otimes\left\vert \Phi\right\rangle \left\langle \Phi\right\vert
_{AB}\right)  U_{\text{SYS},A}^{\dag}\left\vert \Phi\right\rangle _{AB}\\
&  =\sum_{i}\lambda_{i}\left\langle \Phi\right\vert _{AB}U_{\text{SYS}%
,A}\left(  \left\vert \psi_{i}\right\rangle _{\text{SYS}}\left\vert
\Phi\right\rangle _{AB}\otimes\left\langle \psi_{i}\right\vert _{\text{SYS}%
}\left\langle \Phi\right\vert _{AB}\right)  U_{\text{SYS},A}^{\dag}\left\vert
\Phi\right\rangle _{AB}\\
&  =\frac{1}{d^{2}}\sum_{i}\lambda_{i}\text{Tr}_{A}\left\{  U_{\text{SYS}%
,A}\right\}  \left\vert \psi_{i}\right\rangle _{\text{SYS}}\left\langle
\psi_{i}\right\vert _{\text{SYS}}\text{Tr}_{A}\left\{  U_{\text{SYS},A}^{\dag
}\right\}  \\
&  =\frac{1}{d^{2}}\text{Tr}_{A}\left\{  U_{\text{SYS},A}\right\}  \sum
_{i}\lambda_{i}\left\vert \psi_{i}\right\rangle _{\text{SYS}}\left\langle
\psi_{i}\right\vert _{\text{SYS}}\text{Tr}_{A}\left\{  U_{\text{SYS},A}^{\dag
}\right\}  \\
&  =\frac{1}{d^{2}}\text{Tr}_{A}\left\{  U_{\text{SYS},A}\right\}
\rho_{\text{SYS}}\text{Tr}_{A}\left\{  U_{\text{SYS},A}^{\dag}\right\}  ,
\end{align*}
where the fourth equality follows from the above development for pure states.
Thus, the induced map on the chronology-respecting state is as follows (after
renormalization):%
\begin{equation}
\rho\rightarrow\frac{1}{\text{Tr}\left\{  C\rho C^{\dag}\right\}  }C\rho
C^{\dag},\label{eq:PCTC-transform}%
\end{equation}
where%
\[
C\equiv\text{Tr}_{\text{CTC}}\left\{  U\right\}  .
\]
There is always the possibility that the operator $C$ is equivalent to the
null operator, in which case Lloyd \textit{et al}.~suggest that
\textquotedblleft the evolution does not happen\textquotedblright%
\ \cite{LMGGS10}. This result is perhaps strange, suggesting that somehow the
system interacting with the P-CTC\ is annihilated. An explanation for what
could happen resorts to potential imperfections in the unitary interaction.
There is only a paradox for the evolution if the overlap of the CTC qubit with
the final projected state $\left\vert \Phi\right\rangle _{AB}$\ is identically
zero. In practice, evolutions do not occur with arbitrary precision, so that
the P-CTC-assisted circuit magnifies errors dramatically, and unlikely
influences outside the system of interest could intervene before the circuit
can create a paradox.\footnote{All of the results in this paper
should be understood as taking place in a ``hypothetical world,'' where the projection onto the
maximally entangled state $\vert \Phi \rangle$ occurs with certainty. We adopt the nomenclature
``postselection with certainty'' in order to make this point clear.}

P-CTCs allow us to postselect with certainty the outcomes of a generalized
measurement~\cite{Lloyd2010,LMGGS10}. Suppose the generalized measurement
consists of measurement operators $M_{0},...,M_{n-1}$. Suppose that we would
like to postselect the measurement in such a way so that outcome 0 definitely
occurs. We can perform the generalized measurement by appending an ancilla of
dimension at least $n$, in state $\left\vert 0\right\rangle $, to the system,
which we assume to be in a state $\left\vert \psi\right\rangle $. The initial
state is thus $\left\vert \psi\right\rangle \otimes\left\vert 0\right\rangle
$. We then perform a unitary $U_{1}$ that has the following effect:%
\[
U_{1}(\left\vert \psi\right\rangle \otimes\left\vert 0\right\rangle )=\sum
_{k}M_{k}\left\vert \psi\right\rangle \otimes\left\vert k\right\rangle .
\]
(This is the standard construction for a generalized
measurement~\cite{book2000mikeandike}.) Now we do a second unitary $U_{2}$,
from the ancilla to the P-CTC qubit. This unitary is as follows:%
\[
U_{2}=I\otimes\left\vert 0\right\rangle \left\langle 0\right\vert \otimes
I+I\otimes(I-\left\vert 0\right\rangle \left\langle 0\right\vert )\otimes X,
\]
where the third operator in the tensor product acts on the P-CTC qubit. This
construction makes every outcome except $M_{0}$ paradoxical, and measuring the
ancilla in the standard basis postselects so that the resulting state is%
\[
\frac{M_{0}\left\vert \psi\right\rangle }{\left\Vert M_{0}\left\vert
\psi\right\rangle \right\Vert _{2}}\otimes\left\vert 0\right\rangle .
\]
We can postselect on any subset of the measurement outcomes by varying the
projectors in $U_{2}$. For example, we can postselect by accepting any
measurement outcome except $M_{0}$. To do this, we would use the following
unitary as the last one:%
\[
U_{2}=I\otimes\left\vert 0\right\rangle \left\langle 0\right\vert \otimes
X+I\otimes(I-\left\vert 0\right\rangle \left\langle 0\right\vert )\otimes I.
\]

There is an important difference between D-CTCs and P-CTCs that follows
straightforwardly from their definitions. Recall that Deutsch's
self-consistency condition requires that the density matrix of the D-CTC
system after an interaction be equal to the density matrix before the
interaction \cite{PhysRevD.44.3197}. In this way, Deutsch designed D-CTCs
explicitly to replicate exactly the predictions of standard quantum mechanics
in the absence of CTCs. That is, before a CTC comes into existence, or after
it ends, quantum mechanics behaves exactly as usual.

P-CTCs, by contrast, act in a way equivalent to ``postselection with certainty,'' and they
specifically rule out evolutions that lead to a paradox. This implies that the
probabilities of measurement outcomes can be altered \textit{even in the
absence of CTCs}, if CTCs \textit{will} come into existence in the future. In
principle this means that the possibility of CTCs could be tested indirectly,
by looking for deviations from standard quantum probabilities, a fact that was
also pointed out by Hartle in Ref.~\cite{PhysRevD.49.6543}. Needless to say,
it is far from obvious how to do such a test in practice. The bizarre behavior
of nonlinear quantum mechanics would seem to cast doubt that CTCs can exist in
the real world.

We offer a simple example to illustrate the idea in the previous paragraph.
Suppose we have systems $A$, $B$, and $C$, where $C$ is a P-CTC qubit. We
prepare $A$ and $B$ in a maximally entangled state $(\left\vert
00\right\rangle ^{AB}+\left\vert 11\right\rangle ^{AB})/\sqrt{2}$ and measure
$A$ in the Pauli $Z$ basis. Then we perform a CNOT from qubit $B$ to qubit
$C$. This circuit leads to a paradox if the result of measuring $A$ is
$\left\vert 1\right\rangle $, so it must be $\left\vert 0\right\rangle $
(equivalently, one can check that the transformation induced by the P-CTC\ is
$I^{A}\otimes\left\vert 0\right\rangle \left\langle 0\right\vert ^{B}$). But
now consider what happens if we move the preparation and measurement of $AB$
before the P-CTC comes into existence. There are two possibilities:

\begin{enumerate}
\item The usual rules of quantum mechanics apply, and the probabilities of
$\left\vert 0\right\rangle $ and $\left\vert 1\right\rangle $ are equal. If
the result is $\left\vert 1\right\rangle $, we avoid a paradox by magnifying
tiny deviations from the exact unitaries, or other external effects to prevent
the CNOT from happening.

\item The certain postselection forces the measurement result on $A$ to be $\left\vert
0\right\rangle $, even though the P-CTC does not exist yet.
\end{enumerate}

Option 2 is perhaps more natural in this ideal noiseless setting, and it also
matches the qualitative results found by Hartle using path
integrals~\cite{PhysRevD.49.6543}. It is interesting that the system $A$\ does
not have to interact directly with the CTC\ in order for this effect to occur.

\section{Distinguishing Linearly-Independent States with P-CTCs}

\label{sec:P-CTC-distinguish}We begin this section by discussing some simple
examples, and we then prove a general theorem that states that a
P-CTC-assisted circuit can perfectly distinguish an arbitrary set of linearly
independent states and cannot do so if the states are linearly dependent. This
section ends with a discussion of how these circuits act on different
ontological representations of a quantum state.

\begin{figure}[ptb]
\begin{center}
\includegraphics[
natheight=1.159700in,
natwidth=2.572800in,
height=0.6114in,
width=1.5342in
]{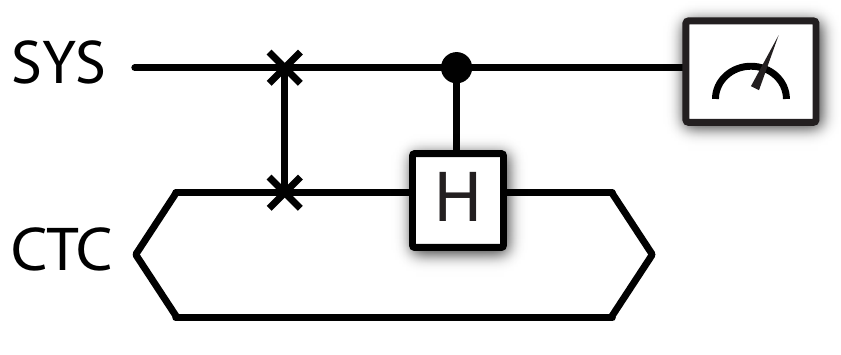}
\end{center}
\caption{A PCTC-assisted circuit that can distinguish when $\left\vert
1\right\rangle $ or $\left\vert +\right\rangle $ is input.}%
\label{fig:B92}%
\end{figure}

Our first circuit in Figure~\ref{fig:B92} distinguishes $\left\vert
1\right\rangle $ from $\left\vert +\right\rangle $ and can thus break the
security of the Bennett-92 protocol for quantum key distribution \cite{b92}.
The circuit consists of a cascade of a SWAP gate followed by a controlled
Hadamard, where%
\begin{align*}
\text{SWAP}  &  \equiv\left\vert 00\right\rangle \left\langle 00\right\vert
+\left\vert 01\right\rangle \left\langle 10\right\vert +\left\vert
10\right\rangle \left\langle 01\right\vert +\left\vert 11\right\rangle
\left\langle 11\right\vert ,\\
\text{C-}H  &  \equiv\left\vert 0\right\rangle \left\langle 0\right\vert
\otimes I+\left\vert 1\right\rangle \left\langle 1\right\vert \otimes H,
\end{align*}
so that%
\[
\left(  \text{C-}H\right)  \left(  \text{SWAP}\right)  =\left\vert
00\right\rangle \left\langle 00\right\vert +\left\vert 01\right\rangle
\left\langle 10\right\vert +\left\vert 1+\right\rangle \left\langle
01\right\vert +\left\vert 1-\right\rangle \left\langle 11\right\vert .
\]
The first qubit upon which the unitary acts is the system qubit, and the
second one is the CTC\ qubit. After tracing over the CTC\ system (as
prescribed in (\ref{eq:PCTC-transform})), we get the following transformation%
\[
C_{\text{B92}}\equiv\left\vert 0\right\rangle \left\langle 0\right\vert
+\left\vert 1\right\rangle \left\langle -\right\vert .
\]
This transformation then gives $\left\vert 0\right\rangle $ if we input
$\left\vert +\right\rangle $ and $\left\vert 1\right\rangle $ if we input
$\left\vert 1\right\rangle $ (after renormalization). Interestingly, this same
circuit distinguishes the antipodal states $\left\vert 0\right\rangle $ and
$\left\vert -\right\rangle $ when assisted by a
D-CTC~\cite{PhysRevLett.102.210402}.

We can generalize the above example to find a P-CTC-assisted circuit that can
distinguish two arbitrary non-orthogonal states. Without loss of generality,
suppose that the two states we are trying to distinguish are $\left\vert
1\right\rangle $ and $\left\vert \phi\right\rangle $ where $\left\vert
\left\langle 1|\phi\right\rangle \right\vert >0$. We would like to build the
following transformation:%
\[
C_{\phi}\equiv\left\vert 0\right\rangle \left\langle 0\right\vert +\left\vert
1\right\rangle \left\langle \phi^{\perp}\right\vert ,
\]
so that $C_{\phi}\left\vert 1\right\rangle =\left\vert 1\right\rangle $ and
$C_{\phi}\left\vert \phi\right\rangle =\left\vert 0\right\rangle $ (after
renormalization). We follow the same prescription as above and exploit the
following unitary $U$:%
\[
U\equiv\left\vert 0\right\rangle \left\langle \phi\right\vert +\left\vert
1\right\rangle \left\langle \phi^{\perp}\right\vert .
\]
We use a cascade of a SWAP\ and a controlled-$U$ where%
\[
\text{C-}U\equiv\left\vert 0\right\rangle \left\langle 0\right\vert \otimes
I+\left\vert 1\right\rangle \left\langle 1\right\vert \otimes U,
\]
so that the cascade $\left(  \text{C-}U\right)  \left(  \text{SWAP}\right)  $
is as follows:%
\[
=\left\vert 00\right\rangle \left\langle 00\right\vert +\left\vert
01\right\rangle \left\langle 10\right\vert +\left\langle \phi|0\right\rangle
\left\vert 10\right\rangle \left\langle 01\right\vert +\left\langle
\phi^{\perp}|0\right\rangle \left\vert 11\right\rangle \left\langle
01\right\vert +\left\langle \phi|1\right\rangle \left\vert 10\right\rangle
\left\langle 11\right\vert +\left\langle \phi^{\perp}|1\right\rangle
\left\vert 11\right\rangle \left\langle 11\right\vert .
\]
After tracing out the CTC\ system, we get%
\begin{align*}
\left\vert 0\right\rangle \left\langle 0\right\vert +\left\langle \phi^{\perp
}|0\right\rangle \left\vert 1\right\rangle \left\langle 0\right\vert
+\left\langle \phi^{\perp}|1\right\rangle \left\vert 1\right\rangle
\left\langle 1\right\vert  &  =\left\vert 0\right\rangle \left\langle
0\right\vert +\left\vert 1\right\rangle \left\langle \phi^{\perp
}|0\right\rangle \left\langle 0\right\vert +\left\vert 1\right\rangle
\left\langle \phi^{\perp}|1\right\rangle \left\langle 1\right\vert \\
&  =\left\vert 0\right\rangle \left\langle 0\right\vert +\left\vert
1\right\rangle \left\langle \phi^{\perp}\right\vert ,
\end{align*}
which is the desired transformation.

The D-CTC-assisted circuit presented in Ref.~\cite{PhysRevLett.102.210402}%
\ for distinguishing BB84 states is not able to distinguish these same states
when assisted by a P-CTC. In fact, the orthogonality relations of the BB84
states remain the same after going through the P-CTC-assisted circuit. The
transformation induced by the circuit in Ref.~\cite{PhysRevLett.102.210402} is
as follows under the P-CTC\ model:%
\[
\left\vert 00\right\rangle \left\langle 00\right\vert +\left\vert
01\right\rangle \left\langle 10\right\vert +\left\vert 10\right\rangle
\left\langle +0\right\vert +\left\vert 11\right\rangle \left\langle
-0\right\vert .
\]
It is perhaps striking that the transformation takes on this form, considering
that the transformation in Ref.~\cite{PhysRevLett.102.210402} takes
$\left\vert 00\right\rangle \rightarrow\left\vert 00\right\rangle $,
$\left\vert 10\right\rangle \rightarrow\left\vert 01\right\rangle $,
$\left\vert +0\right\rangle \rightarrow\left\vert 10\right\rangle $, and
$\left\vert -0\right\rangle \rightarrow\left\vert 11\right\rangle $. One can
check that the output states of the above transformation are as follows (after
renormalization):%
\begin{align*}
\left\vert 00\right\rangle  &  \rightarrow\frac{1}{\sqrt{2}}\left(  \left\vert
00\right\rangle +\left\vert 1+\right\rangle \right)  ,\\
\left\vert 10\right\rangle  &  \rightarrow\frac{1}{\sqrt{2}}\left(  \left\vert
01\right\rangle +\left\vert 1-\right\rangle \right)  ,\\
\left\vert +0\right\rangle  &  \rightarrow\frac{1}{\sqrt{2}}\left(  \left\vert
0+\right\rangle +\left\vert 10\right\rangle \right)  ,\\
\left\vert -0\right\rangle  &  \rightarrow\frac{1}{\sqrt{2}}\left(  \left\vert
0-\right\rangle +\left\vert 11\right\rangle \right)  .
\end{align*}
These states have the same orthogonality relations as the original input
states, and there is thus no improvement in distinguishability. This result
leads us to the main theorem of the next section.

\subsection{P-CTC-assisted circuits produce perfectly distinguishable outputs
for linearly independent input states}

We now state a general theorem regarding state distinguishability and P-CTCs
\footnote{Bennett and Smith both suggested in private
communication~\cite{B09,S10} that it holds true.}. One of our constructions in
the proof has similarities with the general construction in
Ref.~\cite{PhysRevLett.102.210402} for distinguishing an arbitrary set of
non-orthogonal states with a D-CTC-assisted circuit.

\begin{theorem}
\label{thm:P-CTC-state-distinguish}There exists a P-CTC-assisted circuit that
can perfectly distinguish an arbitrary set $\left\{  \left\vert \phi
_{i}\right\rangle \right\}  _{i=1}^{N}$ of linearly independent states, but a
P-CTC-assisted circuit cannot perfectly distinguish a set of linearly
dependent states.
\end{theorem}

\begin{proof}
We present two constructions with the first requiring an $N$-dimensional P-CTC
system, while the second requires only one P-CTC\ qubit.

Our first
construction is similar to the construction in
Ref.~\cite{PhysRevLett.102.210402}\ that uses $N$ D-CTC\ qubits to distinguish
$N$ states. Consider a particular vector $\left\vert \phi_{j}\right\rangle $
in the set $\left\{  \left\vert \phi_{i}\right\rangle \right\}  _{i=1}^{N}$.
Arbitrary superpositions of all the other vectors besides this one outline a
hyperplane of dimension $N-1$ because these states form a linearly independent
set on their own, and let $j$ also refer to this hyperplane. We cannot write
the vector $\left\vert \phi_{j}\right\rangle $ as an arbitrary superposition
of the other states in the set because all the states in it are linearly
independent:%
\[
\left\vert \phi_{j}\right\rangle \neq\sum_{i\neq j}\alpha_{i}\left\vert
\phi_{i}\right\rangle .
\]
For each hyperplane $j$, there is a normal vector $|\widetilde{\phi}%
_{j}\rangle$ such that%
\[
\forall i\neq j:\langle\widetilde{\phi}_{j}|\phi_{i}\rangle=0.
\]
It follows that $\left\vert \langle\widetilde{\phi}_{j}|\phi_{j}%
\rangle\right\vert >0$---were it not so, then $\left\vert \phi_{j}%
\right\rangle $ would lie in hyperplane $j$, which contradicts the assumption
of linear independence.

We would like to have a circuit that implements the
following transformation:%
\[
C\equiv\sum_{j}\left\vert j\right\rangle \langle\widetilde{\phi}_{j}|.
\]
Such a transformation acts as follows (after renormalization) on any state
$\left\vert \phi_{j}\right\rangle $ in the linearly independent set:%
\[
C\left\vert \phi_{j}\right\rangle =\left\vert j\right\rangle .
\]
The output of this transformation is then distinguishable with a von Neumann
measurement.

We now explicitly construct a unitary that implements the above
transformation after tracing over the CTC\ system. It is a cascade of a qudit
SWAP\ gate and a particular controlled unitary gate (a generalization of our
examples from before). The qudit SWAP\ gate is as follows:%
\[
\sum_{j,k}\left\vert j\right\rangle \left\langle k\right\vert \otimes
\left\vert k\right\rangle \left\langle j\right\vert ,
\]
and the controlled unitary gate is%
\[
\sum_{l}\left\vert l\right\rangle \left\langle l\right\vert \otimes U_{l},
\]
where we choose each unitary $U_{l}$ above so that%
\[
\left\langle l\right\vert U_{l}=\langle\tilde{\phi}_{l}|,
\]
and its action on other basis states besides $\left\vert l\right\rangle $ is
not important. Then the cascade of these gates gives%
\begin{align*}
\left(  \sum_{l}\left\vert l\right\rangle \left\langle l\right\vert \otimes
U_{l}\right)  \left(  \sum_{j,k}\left\vert j\right\rangle \left\langle
k\right\vert \otimes\left\vert k\right\rangle \left\langle j\right\vert
\right)   &  =\sum_{j,k,l}\left\vert l\right\rangle \left\langle
l|j\right\rangle \left\langle k\right\vert \otimes U_{l}\left\vert
k\right\rangle \left\langle j\right\vert \\
&  =\sum_{j,k}\left\vert j\right\rangle \left\langle k\right\vert \otimes
U_{j}\left\vert k\right\rangle \left\langle j\right\vert .
\end{align*}
We finally trace out the CTC\ system to determine the actual transformation on
the chronology-respecting system:%
\begin{align*}
\sum_{j,k}\left\vert j\right\rangle \left\langle k\right\vert \left\langle
j\right\vert U_{j}\left\vert k\right\rangle  &  =\sum_{j,k}\left\vert
j\right\rangle \left\langle j\right\vert U_{j}\left\vert k\right\rangle
\left\langle k\right\vert \\
&  =\sum_{j}\left\vert j\right\rangle \left\langle j\right\vert U_{j}\sum
_{k}\left\vert k\right\rangle \left\langle k\right\vert \\
&  =\sum_{j}\left\vert j\right\rangle \left\langle j\right\vert U_{j}\\
&  =\sum_{j}\left\vert j\right\rangle \langle\widetilde{\phi}_{j}|.
\end{align*}
This last step proves that the construction gives the desired transformation.

There is another construction which can accomplish the same task with just one
P-CTC\ qubit. By choosing the POVM that distinguishes any set of linearly
independent states \cite{C98}, and ruling out result $M_{0}$ (which is
\textquotedblleft I do not know\textquotedblright), we can construct a
P-CTC-assisted circuit of the form in Section~\ref{sec:review}\ that can
distinguish any set of linearly independent states with just one P-CTC\ qubit.
This circuit performs the transformation $\left\vert \phi_{i}\right\rangle
\rightarrow\left\vert i\right\rangle $ so that the states at the output of the
circuit are perfectly distinguishable with a von Neumann measurement.

We now
prove the other part of Theorem~\ref{thm:P-CTC-state-distinguish}---that
linearly-dependent states are not perfectly distinguishable with a
P-CTC-assisted circuit. Consider an arbitrary unitary $U$ that acts on the
chronology-respecting system and the CTC\ system. We can decompose it as
follows:%
\[
U=\sum_{j,k}A_{j,k}\otimes\left\vert j\right\rangle \left\langle k\right\vert
,
\]
with respect to some basis $\left\{  \left\vert j\right\rangle \right\}
$\ for the CTC\ system. Tracing out the CTC\ system gives the transformation
that the P-CTC-assisted circuit induces%
\[
C=\text{Tr}_{\text{CTC}}\left\{  U\right\}  =\sum_{j}A_{j,j}.
\]
Now suppose that a P-CTC\ can distinguish a state $\left\vert \phi
_{0}\right\rangle $ from $\left\vert \phi_{1}\right\rangle $ in the sense that
$C\left\vert \phi_{0}\right\rangle =\left\vert 0\right\rangle $ and
$C\left\vert \phi_{1}\right\rangle =\left\vert 1\right\rangle $ after
renormalization. Then consider a linearly dependent state $\left\vert
\psi\right\rangle $ where we can write $\left\vert \psi\right\rangle
=\alpha\left\vert \phi_{0}\right\rangle +\beta\left\vert \phi_{1}\right\rangle
$ for some $\alpha,\beta\neq0$. Then, by linearity of the transformation $C$
before renormalization, it follows that%
\[
C\left\vert \psi\right\rangle =C\left(  \alpha\left\vert \phi_{0}\right\rangle
+\beta\left\vert \phi_{1}\right\rangle \right)  =\alpha C\left\vert \phi
_{0}\right\rangle +\beta C\left\vert \phi_{1}\right\rangle =\alpha
\mathcal{N}_{0}\left\vert 0\right\rangle +\beta\mathcal{N}_{1}\left\vert
1\right\rangle ,
\]
where $\mathcal{N}_{0}$ and $\mathcal{N}_{1}$ are non-zero normalization
constants. After renormalization, this state is not distinguishable from
$\left\vert 0\right\rangle $ or $\left\vert 1\right\rangle $ by any
measurement. This proof generalizes easily so that any P-CTC\ transformation
would not be able to distinguish a general set of linearly dependent states.
\end{proof}

As an afterthought, the above theorem demonstrates that the power of a
P-CTC-assisted circuit is rather limited in comparison to a D-CTC-assisted
one. A D-CTC-assisted circuit can violate the Holevo bound
\cite{PhysRevLett.102.210402}, but a P-CTC-assisted one can never violate it
because one can never have more than $N$ linearly independent states in $N$
dimensions. Of course, if the receiver has access to a P-CTC, that will raise
the classical capacity of certain channels, since it increases the ability to
distinguish states beyond that of ordinary quantum mechanics. The theorem also
implies that a P-CTC-assisted circuit cannot break the security of the BB84
\cite{bb84}\ or SARG04 \cite{sarg04}\ protocols for quantum key distribution,
though a P-CTC will increase the power of the eavesdropper to a certain
degree. These results might lend further credence to the belief that P-CTCs
are a more reasonable model of time travel because their information
processing abilities are not as striking as those of D-CTCs (even though they
still violate the uncertainty principle).

\subsection{The \textquotedblleft Linearity Trap\textquotedblright\ for
Labeled Mixtures}

The operation of a D-CTC-assisted circuit on a labeled mixture of states is
controversial \cite{PhysRevLett.103.170502,RM10,CM10} and can lead to
dramatically different conclusions depending on how one interprets such a
labeled mixture. A similar phenomenon happens with P-CTCs as we discuss below.

Let us consider a general ensemble $\left\{  \left(  p\left(  x\right)
,\left\vert \phi_{x}\right\rangle \right)  \right\}  $\ of non-orthogonal,
linearly independent states. In linear quantum mechanics, this ensemble has a
one-to-one correspondence with the following labeled mixture:%
\begin{equation}
\sum_{x}p\left(  x\right)  \left\vert x\right\rangle \left\langle x\right\vert
^{X}\otimes\left\vert \phi_{x}\right\rangle \left\langle \phi_{x}\right\vert
^{A}, \label{eq:labeled-mixture}%
\end{equation}
where the states $\{\left\vert x\right\rangle ^{X}\}$ are an orthonormal set.
Suppose the preparer holds on to the $X$ label and sends the $A$ system
through the transformation from Theorem~\ref{thm:P-CTC-state-distinguish}. A
first way to renormalize would be to act with the P-CTC\ transformation on
each state $\left\vert \phi_{x}\right\rangle $ in the ensemble and renormalize
each resulting state. This procedure assumes that the classical labeling
information is available, in principle. This process leads to the output
ensemble $\{(p\left(  x\right)  ,\left\vert x\right\rangle ^{A})\}$, which has
a one-to-one correspondence with the following labeled mixture:%
\begin{equation}
\sum_{x}p\left(  x\right)  \left\vert x\right\rangle \left\langle x\right\vert
^{X}\otimes\left\vert x\right\rangle \left\langle x\right\vert ^{A},
\label{eq:individual-renormalize}%
\end{equation}
so that the systems on $X$ and $A$ are now classically correlated according to
the distribution $p\left(  x\right)  $.

Another method for renormalization leads to a drastically different result.
Considering the labeled mixture as a \textquotedblleft true density
matrix\textquotedblright\ and acting on this state with the transformation
from Theorem~\ref{thm:P-CTC-state-distinguish}\ gives%
\begin{align*}
&  \sum_{j}\left\vert j\right\rangle \langle\tilde{\phi}_{j}|^{A}\left(
\sum_{x}p\left(  x\right)  \left\vert x\right\rangle \left\langle x\right\vert
^{X}\otimes\left\vert \phi_{x}\right\rangle \left\langle \phi_{x}\right\vert
^{A}\right)  \sum_{j^{\prime}}|\widetilde{\phi}_{j^{\prime}}\rangle\langle
j^{\prime}|^{A}\\
&  =\sum_{x}p\left(  x\right)  \left\vert x\right\rangle \left\langle
x\right\vert ^{X}\otimes\sum_{j,j^{\prime}}\left\vert j\right\rangle
\langle\widetilde{\phi}_{j}|\phi_{x}\rangle\langle\phi_{x}|\widetilde{\phi
}_{j^{\prime}}\rangle\langle j^{\prime}|^{A}\\
&  =\sum_{x}p\left(  x\right)  \left\vert x\right\rangle \left\langle
x\right\vert ^{X}\otimes\sum_{j}|\langle\widetilde{\phi}_{j}|\phi_{x}%
\rangle|^{2}\ \left\vert j\right\rangle \langle j|^{A}\\
&  =\sum_{x,j}p\left(  x\right)  |\langle\widetilde{\phi}_{j}|\phi_{x}%
\rangle|^{2}\left\vert x\right\rangle \left\langle x\right\vert ^{X}%
\otimes\left\vert j\right\rangle \langle j|^{A}\\
&  =\sum_{x}p\left(  x\right)  |\langle\tilde{\phi}_{x}|\phi_{x}\rangle
|^{2}\left\vert x\right\rangle \left\langle x\right\vert ^{X}\otimes\left\vert
x\right\rangle \left\langle x\right\vert ^{A}.%
\end{align*}
After renormalization, the state is as follows:%
\begin{equation}
\sum_{x}q\left(  x\right)  \left\vert x\right\rangle \left\langle x\right\vert
^{X}\otimes\left\vert x\right\rangle \left\langle x\right\vert ^{A},
\label{eq:strange-state}%
\end{equation}
where%
\[
q\left(  x\right)  \equiv\frac{1}{\sum_{x}p\left(  x\right)  |\langle
\tilde{\phi}_{x}|\phi_{x}\rangle|^{2}}p\left(  x\right)  |\langle\tilde{\phi
}_{x}|\phi_{x}\rangle|^{2}.
\]
The systems on $X$ and $A$ are classically correlated again, but the
distribution for the correlation can be drastically different if the overlap
$|\langle\tilde{\phi}_{x}|\phi_{x}\rangle|^{2}$ is not uniform over $x$. The
interpretation of the above result is bizarre:\ the P-CTC-assisted circuit
changes the original probabilities of the states in the mixture, in spite of
the fact that the preparer generated these probabilities well before the
P-CTC\ even came into existence.

Let us examine a third scenario. Consider the following purification of the
state in (\ref{eq:labeled-mixture}):%
\[
\sum_{x}\sqrt{p\left(  x\right)  }\left\vert x\right\rangle ^{X}\left\vert
x\right\rangle ^{X^{\prime}}\left\vert \phi_{x}\right\rangle ^{A}.
\]
Suppose that Alice sends the $A$ system through the P-CTC-assisted circuit.
Acting on this state with the transformation from
Theorem~\ref{thm:P-CTC-state-distinguish} gives%
\begin{align*}
\sum_{j}\left\vert j\right\rangle \langle\tilde{\phi}_{j}|^{A}\sum_{x}%
\sqrt{p\left(  x\right)  }\left\vert x\right\rangle ^{X}\left\vert
x\right\rangle ^{X^{\prime}}\left\vert \phi_{x}\right\rangle ^{A}  &
=\sum_{x,j}\sqrt{p\left(  x\right)  }\left\vert x\right\rangle ^{X}\left\vert
x\right\rangle ^{X^{\prime}}\left\vert j\right\rangle ^{A}\langle\tilde{\phi
}_{j}|\phi_{x}\rangle^{A}\\
&  =\sum_{x}\sqrt{p\left(  x\right)  }\langle\tilde{\phi}_{x}|\phi_{x}%
\rangle^{A}\left\vert x\right\rangle ^{X}\left\vert x\right\rangle
^{X^{\prime}}\left\vert x\right\rangle ^{A}.
\end{align*}
Renormalizing the last line above leads to the following state:%
\begin{equation}
\frac{1}{\mathcal{N}}\sum_{x}\sqrt{p\left(  x\right)  }\langle\tilde{\phi}%
_{x}|\phi_{x}\rangle^{A}\left\vert x\right\rangle ^{X}\left\vert
x\right\rangle ^{X^{\prime}}\left\vert x\right\rangle ^{A},
\label{eq:purified-calculation}%
\end{equation}
where%
\[
\mathcal{N}\equiv\sqrt{\sum_{x}p\left(  x\right)  |\langle\tilde{\phi}%
_{x}|\phi_{x}\rangle|^{2}}.
\]
The coefficients $\langle\tilde{\phi}_{x}|\phi_{x}\rangle^{A}$ can generally
be complex, and this state is different from the other outcomes illustrated
above because there is quantum interference. Though, this state is a
purification of the state in (\ref{eq:strange-state}), so that the resulting
state is the same as in (\ref{eq:strange-state}) if we discard the system
$X^{\prime}$.

What should we make of these differing results? In standard quantum mechanics,
there are three concepts that are indistinguishable from each other:

\begin{enumerate}
\item An ensemble $\{p\left(  x\right)  ,\left\vert \phi_{x}\right\rangle
\left\langle \phi_{x}\right\vert \}$ of pure states (classical ignorance);

\item An entangled state $\left\vert \psi\right\rangle \left\langle
\psi\right\vert ^{AB}$ where subsystem $B$ is assumed to be inaccessible;

\item A density matrix $\rho$.
\end{enumerate}

The first is what d'Espagnat called a \textquotedblleft proper
mixture,\textquotedblright\ and the second is what he called an
\textquotedblleft improper mixture\textquotedblright\ \cite{BDE06}. The
density matrix is a mathematical object introduced to summarize the observable
consequences of either of the other two. One could imagine such a thing as a
\textquotedblleft true density matrix\textquotedblright---an intrinsically
mixed state that does not represent either classical ignorance or
entanglement---though it is not clear what such an object would mean,
physically. However, some researchers suggest that density matrices, rather
than pure states, should be the fundamental objects in quantum theory. For
instance, mixed states arise naturally in Deutsch's approach to
CTCs~\cite{PhysRevD.44.3197}.

These three different ontological representations of a quantum state are all
indistinguishable in standard quantum mechanics because it is a linear theory.
But in a nonlinear version of quantum mechanics (as we get using either D-CTCs
or P-CTCs) we have no reason to expect them to behave the same, and they do
not. This is the major criticism that Ref.~\cite{CM10}\ levies against
Ref.~\cite{PhysRevLett.103.170502}. Bennett \textit{et al}.~use
\textquotedblleft labeled mixtures\textquotedblright\ to represent classical
ensembles, which is not necessarily justifiable in nonlinear quantum mechanics.

So what should be done? If the \textquotedblleft labeled
mixture\textquotedblright\ represents classical ignorance about how the system
was prepared, that implies that, in principle, information exists which would
specify a pure state. One should therefore apply the calculation separately to
each state $\left\vert \phi_{x}\right\rangle $ in the ensemble, and then
combine them in a new ensemble. If the output state is random (e.g., the
result of a measurement), one gets the probabilities of the new ensemble using the
Bayes rule. This is the first approach in (\ref{eq:individual-renormalize}).

If the mixture is really part of an entangled state, one should apply the
calculation to the purification of the state and then trace out the
inaccessible subsystem as in (\ref{eq:purified-calculation}). This procedure
will, in general, give a different answer, as (\ref{eq:purified-calculation}) demonstrates.

Finally, if there are such objects as \textquotedblleft true density
matrices,\textquotedblright\ one can calculate with them directly. This is
what we do in (\ref{eq:strange-state}), and it gives the same answer as
tracing out the reference system $X^{\prime}$\ of the purified state in
(\ref{eq:purified-calculation}). Also, Bennett \textit{et al}.~assume that a
labeled mixture of states is a \textquotedblleft true density
matrix,\textquotedblright\ and this assumption is what leads them to conclude
that D-CTCs are \textquotedblleft impotent\textquotedblright\ in
Refs.~\cite{PhysRevLett.103.170502,BLSS10}.

In summary, the first ontological representation of a quantum state as a
proper mixture leads to a dramatically different conclusion for P-CTC-assisted
circuits than the second and third ontological representations (which both
lead to the same conclusion).\ We also note that it is the same with
D-CTC-assisted circuits (the first representation leads to differing
conclusions than the second/third), but the states output by a D-CTC-assisted
circuit are different from those output by a P-CTC-assisted one.

\section{P-CTCs can help solve hard problems}

\label{sec:P-CTC-compute}Lloyd \textit{et al}.~prove that the computational
power of quantum computers and P-CTCs is equivalent to PP (\textquotedblleft
Probabilistic Polynomial time\textquotedblright)~\cite{Lloyd2010}, whereas the
computational power of Deutsch's CTCs \cite{PhysRevD.44.3197}\ is
PSPACE~\cite{ScottAaronson02082009}. The proof is simple---since we can
simulate any postselected measurement with P-CTCs, and can simulate P-CTCs
with postselected measurements, the two paradigms have equivalent
computational power. Since Aaronson proved that quantum mechanics with
postselection has computational power PP~\cite{A05}, PCTCs indeed have
computational power equivalent to that of PP. PP is a rather powerful
computational class, including (for example) all problems in NP. It is known
to be contained in PSPACE, however, and is generally believed to be less powerful.

It is instructive to outline explicit P-CTC-assisted circuits that illustrate
the power of P-CTCs. In the next few sections, we give explicit P-CTC-assisted
circuits that can factor efficiently, can solve any problem in the
intersection of NP\ and co-NP, and can probabilistically solve any problem in
NP. All of these circuits use just one P-CTC qubit. The structure of these
circuits draws on ideas from the algorithms in Ref.~\cite{fpl2003brun}, and
are closely related to the construction of Aaronson~\cite{A05} in his proof
that NP $\subseteq$ PostBQP.

\subsection{FACTORING}

\begin{figure}[ptb]
\begin{center}
\includegraphics[
natheight=2.786400in,
natwidth=5.453500in,
height=1.567in,
width=3.039in
]{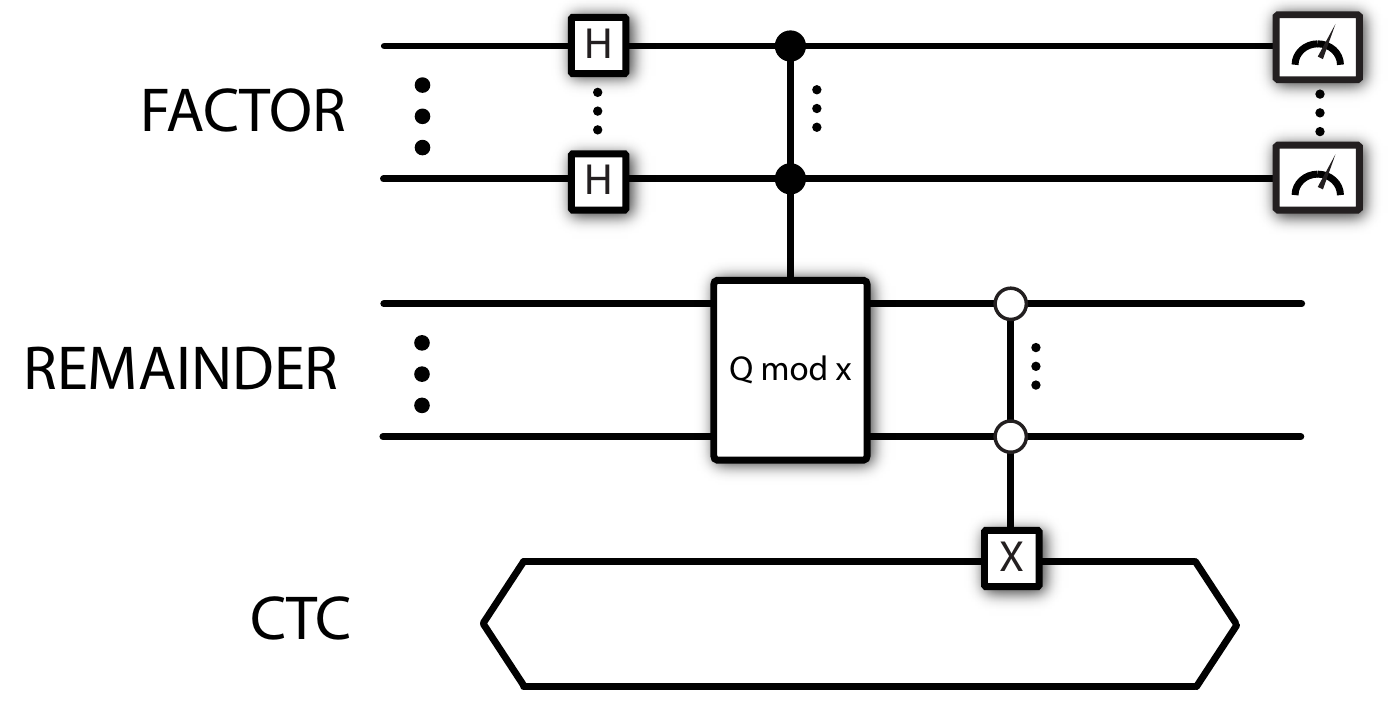}
\end{center}
\caption{A PCTC-assisted circuit that can factor an integer efficiently. The
circuit exploits the ability of a P-CTC\ to eliminate incorrect solutions by
making them paradoxical.}%
\label{fig:factoring}%
\end{figure}

Let $Q$ be the number to factor and suppose that it is $N$ bits long (so that
$N\approx\log(Q)$). The P-CTC-assisted circuit consists of a CTC\ qubit and
two $N$-qubit registers:\ a {\small REMAINDER} register and a {\small FACTOR}
register. The steps of the circuit are as follows (depicted in
Figure~\ref{fig:factoring}):

\begin{enumerate}
\item Initialize the {\small FACTOR} and {\small REMAINDER}\ registers to
$\left\vert 0\right\rangle $. Apply Hadamard gates to all the qubits in the
{\small FACTOR} register. This first set of Hadamards is equivalent to the
following unitary:%
\[
U_{1}\equiv\left(  H^{\otimes N}\right)  _{\text{F}},
\]
where \textquotedblleft F\textquotedblright\ denotes the {\small FACTOR} register.

\item Act with a controlled unitary that calculates the modulo operation on
the {\small FACTOR} register and places it in the {\small REMAINDER} register:%
\[
U_{2}\equiv\sum_{j=0}^{2^{N}-1}\left(  U_{j}\right)  _{\text{R}}%
\otimes\left\vert j\right\rangle \left\langle j\right\vert _{\text{F}}.
\]
In the above, \textquotedblleft R\textquotedblright\ indicates the
{\small REMAINDER}\ register, and $U_{j}$ is some unitary chosen so that%
\[
U_{j}\left\vert 0\right\rangle =\left\{
\begin{array}
[c]{cc}%
\left\vert 1\right\rangle  & \text{if\ }j\in\left\{  0,Q\right\} \\
\left\vert Q\operatorname{mod}j\right\rangle  & \text{else}%
\end{array}
\right.  .
\]
If $j$ divides $Q$ and $j$ is not equal to $0$ or $Q$, then the
{\small REMAINDER} register contains $0$. Otherwise, it contains a nonzero
number, the remainder of $Q/j$.

\item Apply the following controlled unitary from the {\small REMAINDER}
register to the CTC register:%
\[
U_{3}\equiv\left\vert 0\right\rangle \left\langle 0\right\vert _{\text{R}%
}\otimes I_{\text{F}}\otimes I_{\text{CTC}}+(I-\left\vert 0\right\rangle
\left\langle 0\right\vert )_{\text{R}}\otimes I_{\text{F}}\otimes
X_{\text{CTC}}.
\]

\item Measure the {\small FACTOR}\ register, and let the CTC qubits continue
through the CTC.
\end{enumerate}

We can verify that the P-CTC-assisted circuit is behaving as it should by
considering the induced transformation (as in (\ref{eq:PCTC-transform})). The
cascade of $U_{1}$, $U_{2}$, and $U_{3}$ is the following unitary:%
\[
\sum_{j=0}^{2^{N}-1}\left[  \left\vert 0\right\rangle \left\langle
0\right\vert U_{j}\right]  _{\text{R}}\otimes\left[  \left\vert j\right\rangle
\left\langle j\right\vert H^{\otimes N}\right]  _{\text{F}}\otimes
I_{\text{CTC}}+\sum_{j=0}^{2^{N}-1}\left[  \left(  I-\left\vert 0\right\rangle
\left\langle 0\right\vert \right)  U_{j}\right]  _{\text{R}}\otimes\left[
\left\vert j\right\rangle \left\langle j\right\vert H^{\otimes N}\right]
_{\text{F}}\otimes X_{\text{CTC}}.
\]
Tracing over the CTC\ qubit gives the induced transformation:%
\[
\sum_{j=0}^{2^{N}-1}\left[  \left\vert 0\right\rangle \left\langle
0\right\vert U_{j}\right]  _{\text{R}}\otimes\left[  \left\vert j\right\rangle
\left\langle j\right\vert H^{\otimes N}\right]  _{\text{F}}.
\]
Applying this transformation to a {\small FACTOR} and {\small REMAINDER}%
\ register both initialized to $\left\vert 0\right\rangle $ gives the
following state:%
\[
\left\vert 0\right\rangle _{\text{R}}\sum_{j\ :\ Q\operatorname{mod}j=0\wedge
j\neq Q\wedge j\neq0}\left\vert j\right\rangle _{\text{F}}\text{,}%
\]
where we see that the effect of the last controlled gate in
Figure~\ref{fig:B92}\ is to eliminate all of the invalid answers or the ones
for which $j\in\left\{  0,Q\right\}  $\ by making these possibilities
paradoxical. Measuring the {\small FACTOR}\ register then returns a factor of
$Q$. The algorithm fails in the case where $Q$ is prime. But since primality
can be checked efficiently, we just assume that we only use the algorithm with
a composite (non-prime) $Q$.

If the CTC qubits produce a number $\left\vert j\right\rangle $ that is not a
factor of $Q$, then a NOT is applied to the CTC qubit. This would produce a
paradox, which is forbidden for P-CTCs. Because the initial state coming out
of the CTC contains components including all numbers $\left\vert
j\right\rangle $, those components that are factors of $Q$ have their
probabilities magnified, and all others have their probabilities suppressed.
The circuit uses the grandfather paradox such that the only histories that
return a factor are self-consistent. This idea is the same as that in
Ref.~\cite{fpl2003brun} and also exploited by Aaronson to illustrate the power
of postselected quantum computation~\cite{A05}.

\subsection{Decision Problems in the Intersection of NP and co-NP}

\begin{figure}[ptb]
\begin{center}
\includegraphics[
natheight=1.987300in,
natwidth=5.927400in,
height=1.1718in,
width=3.4411in
]
{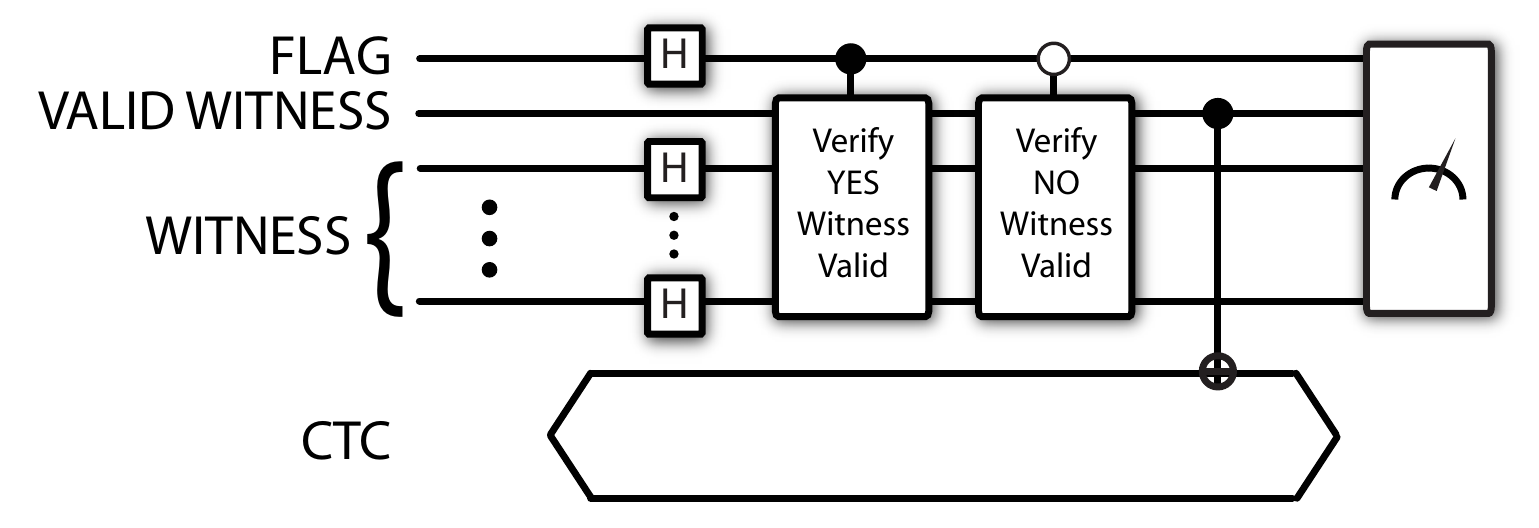}
\end{center}
\caption{A PCTC-assisted circuit for determining a valid witness for a
decision problem in the intersection of NP and co-NP.}%
\label{fig:NP-and-co-NP}%
\end{figure}

A decision problem lying in the intersection of NP and co-NP is one for which
there is a short witness for both a YES\ answer and a NO answer. We can solve
any decision problem in this complexity class with a P-CTC-assisted circuit.
The idea here is essentially the same as in the P-CTC-assisted factoring
algorithm, only there are now two parts. Suppose that for a particular
decision problem both witnesses can be represented using no more than $N$
bits. The P-CTC-assisted circuit consists of four quantum registers: a
{\small FLAG}\ qubit, a {\small VALID\ WITNESS}\ qubit, a {\small WITNESS}%
\ register with $N$ qubits, and a CTC\ qubit. It operates as follows (depicted
in Figure~\ref{fig:NP-and-co-NP}):

\begin{enumerate}
\item Initialize the {\small FLAG}\ qubit, the {\small VALID\ WITNESS}\ qubit,
and the {\small WITNESS}\ register to $\left\vert 0\right\rangle $. Apply
Hadamard gates to the {\small FLAG}\ qubit and to the $N$ qubits in the
{\small WITNESS} register. The flag qubit being equal to 1 means that the
answer is YES. The {\small FLAG}\ qubit being equal to 0 means the answer is NO.

\item Conditioned on the {\small FLAG}\ qubit being equal to 1, the answer is
(claimed to be) YES and the remaining $N$ qubits hold a witness $\left\vert
y\right\rangle $. The {\small FLAG} qubit acts as a control bit. If
{\small FLAG} = 1, then pass the $N$ qubits of the witness, plus the
{\small VALID\ WITNESS}\ qubit in the state $\left\vert 0\right\rangle
_{\text{V}}$, through a circuit that verifies whether the witness is valid:%
\[
\left\vert 0\right\rangle _{\text{V}}\left\vert y\right\rangle _{\text{W}%
}\rightarrow\left\vert j\right\rangle _{\text{V}}\left\vert y\right\rangle
_{\text{W}},
\]
where \textquotedblleft V\textquotedblright\ denotes the
{\small VALID\ WITNESS}\ qubit, \textquotedblleft W\textquotedblright\ denotes
the {\small WITNESS}\ register, $j=0$ if the witness is valid, and $j=1$ otherwise.

\item Conditioned on the {\small FLAG}\ qubit being equal to 0, the answer is
(claimed to be) NO and the remaining $N$ qubits hold a witness $\left\vert
n\right\rangle $. The {\small FLAG} qubit again acts as a control bit. If
{\small FLAG} = 0, then pass the $N$ qubits of the witness, plus the
{\small VALID\ WITNESS}\ qubit in the state $\left\vert 0\right\rangle
_{\text{V}}$, through a circuit that verifies whether the witness is valid:%
\[
\left\vert 0\right\rangle _{\text{V}}\left\vert n\right\rangle _{\text{W}%
}\rightarrow\left\vert j\right\rangle _{\text{V}}\left\vert n\right\rangle
_{\text{W}},
\]
where $j=0$ if the witness is valid, and $j=1$ otherwise.

\item Apply a CNOT from the {\small VALID\ WITNESS}\ qubit holding $\left\vert
j\right\rangle $ to the CTC qubit.

\item Measure the {\small FLAG}\ qubit, the {\small VALID\ WITNESS}\ qubit,
and the {\small WITNESS}\ register. The measurement results give an answer to
the decision problem (in the {\small FLAG}\ qubit and in the
{\small VALID\ WITNESS}\ qubit) plus the witness.
\end{enumerate}

The reasoning that this algorithm works is essentially the same as that for
factoring. The last CNOT\ gate makes a paradox out of any scenario in which
the {\small VALID\ WITNESS}\ qubit is equal to one, thus eliminating the
possibilities for which the witness is invalid.

\subsection{SAT}

A satisfiability (SAT) decision problem tries to determine if there exists a
satisfying solution for a Boolean formula (one which makes the formula
evaluate to TRUE). We now show that we can probabilistically solve a
SAT\ decision problem, which implies that we can probabilistically solve any
problem in NP because SAT is NP-complete~\cite{C71}. Suppose that we want to
solve SAT on $N$ bits. There is a Boolean function $f\left(  x_{1}%
,\ldots,x_{N}\right)  $ defined by a formula that can be evaluated
efficiently. We want to know if there are values of $x_{1},\ldots,x_{N}$ that
make $f\left(  x_{1},\ldots,x_{N}\right)  =1$, and if so, we would like to
have a satisfying assignment. (The latter is not necessary for a decision
problem, of course, but we get it for free.)

The P-CTC-assisted circuit acts on four different quantum registers:\ a
{\small FLAG} qubit, a {\small VALID\ WITNESS}\ qubit, an $N$-qubit
{\small WITNESS}\ register, and a CTC qubit. If the {\small FLAG} qubit is
equal to 1, the function has a satisfying assignment, and if it is 0, it does
not.\ The circuit has the following steps:

\begin{enumerate}
\item Initialize the {\small FLAG}\ qubit, the {\small VALID\ WITNESS}\ qubit,
and the {\small WITNESS}\ register to $\left\vert 0\right\rangle $. Apply
Hadamard gates to the {\small FLAG} qubit and the $N$-qubit {\small WITNESS}\ register.

\item Conditioned on the {\small FLAG} qubit being equal to 1, the answer is
(claimed to be) YES, and the {\small WITNESS}\ register holds a satisfying
assignment $\left\vert x_{1},\ldots,x_{N}\right\rangle _{\text{W}}$. The
{\small FLAG} qubit acts as a control bit. If {\small FLAG} = 1, then pass the
$N$-qubit {\small WITNESS}\ register, plus the {\small VALID\ WITNESS}\ qubit
in the state $\left\vert 0\right\rangle _{\text{V}}$, through a circuit that
calculates%
\[
\left\vert x_{1},\ldots,x_{N}\right\rangle _{\text{W}}\left\vert
0\right\rangle _{\text{V}}\rightarrow\left\vert x_{1},\ldots,x_{N}%
\right\rangle _{\text{W}}\left\vert j\right\rangle _{\text{V}},
\]
where $j=\lnot f\left(  x_{1},\ldots,x_{N}\right)  $. (That is, $j=0$ if
$x_{1},\ldots,x_{N}$ is satisfying, and $j=1$ if not.)

\item Conditioned on the {\small FLAG} qubit being equal to 0, the answer is
(claimed to be) NO, and we require that the $N$-qubit {\small WITNESS}%
\ register hold all zeros: $\left\vert 00\ldots0\right\rangle $. Apply the
following controlled-unitary:%
\[
I_{\text{F}}\otimes\left\vert 0\ldots0\right\rangle \left\langle
0\ldots0\right\vert _{\text{W}}\otimes I_{\text{V}}+\left\vert 1\right\rangle
\left\langle 1\right\vert _{\text{F}}\otimes\left(  I-\left\vert
0\ldots0\right\rangle \left\langle 0\ldots0\right\vert \right)  _{\text{W}%
}\otimes I_{\text{V}}+\left\vert 0\right\rangle \left\langle 0\right\vert
_{\text{F}}\otimes\left(  I-\left\vert 0\ldots0\right\rangle \left\langle
0\ldots0\right\vert \right)  _{\text{W}}\otimes X_{\text{V}}.
\]
The {\small VALID\ WITNESS}\ qubit now holds $\left\vert 1\right\rangle
_{\text{V}}$ if the qubits in the $N$-qubit {\small WITNESS}\ register are not
all zeros, and $\left\vert 0\right\rangle _{\text{V}}$ if they are.

\item Apply a CNOT from the {\small VALID\ WITNESS}\ qubit $\left\vert
j\right\rangle _{\text{V}}$ to the CTC qubit.

\item Measure all the ancillas.
\end{enumerate}

There are two cases:

\begin{enumerate}
\item If the function has no satisfying assignment, then the only
non-paradoxical output is all zeros (including the flag bit): $\left\vert
0\right\rangle _{\text{F}}\left\vert 0\right\rangle _{\text{V}}\left\vert
00\ldots0\right\rangle _{\text{W}}$. This outcome occurs with probability one
in this case.

\item If the function has $m$ satisfying assignments, then there are $m+1$
non-paradoxical results: the $m$ satisfying assignments, plus the all zero
state $\left\vert 0\right\rangle _{\text{F}}\left\vert 0\right\rangle
_{\text{V}}\left\vert 00\ldots0\right\rangle _{\text{W}}$. These $m+1$ results
occur with equal probability.
\end{enumerate}

So in case 1, the correct answer (NO) always occurs, and in case 2, a
satisfying assignment (YES) occurs with probability%
\[
\frac{m}{m+1}\geq1/2,
\]
and the false answer (NO) occurs with probability%
\[
\frac{1}{m+1}\leq1/2.
\]

To improve the probabilities, we can replicate some of these steps while still
using just one CTC\ qubit. Replicate steps 1-3 $k$ times on $k$ copies of all
of the above registers (except for the CTC\ qubit). So we now get $k$
different flag bits and (potentially) satisfying assignments. We then do the
following unitary from the $k$ {\small VALID\ WITNESS} qubits to the P-CTC qubit:%

\[
\left\vert 00\ldots0\right\rangle \left\langle 00\ldots0\right\vert \otimes
I+(I-\left\vert 00\ldots0\right\rangle \left\langle 00\ldots0\right\vert
)\otimes X.
\]
In case 1, we get the result $\left\vert 0\right\rangle _{\text{F}}\left\vert
0\right\rangle _{\text{V}}\left\vert 00\ldots0\right\rangle _{\text{W}}$ every
time. In case 2, we get the wrong answer $\left\vert 0\right\rangle
_{\text{F}}\left\vert 0\right\rangle _{\text{V}}\left\vert 00\ldots
0\right\rangle _{\text{W}}$ every time only with probability:%
\[
\frac{1}{(m+1)^{k}}\leq\frac{1}{2^{k}}.
\]
We can make this failure probability as small as we like with only logarithmic overhead.

\section{Conclusion}

Prior research has shown that closed timelike curves operating according to
Deutsch's model can have dramatic consequences for computation and information
processing \cite{ScottAaronson02082009,PhysRevLett.102.210402}\ if one
operates on \textquotedblleft proper\textquotedblright\ mixtures of quantum
states \cite{BDE06}. Lloyd \textit{et al}.~then showed that postselected
closed timelike curves have computational power equivalent to the complexity
class PP~\cite{Lloyd2010,LMGGS10}, by exploiting a result of Aaronson on
postselected quantum computation \cite{A05}.

In this paper, we showed how to implement any postselected operation with certainty with just
one P-CTC\ qubit, and we discussed an important difference between D-CTCs and
P-CTCs in which the future existence of a P-CTC\ could affect the
probabilistic results of a present experiment by creating a paradox for
particular outcomes. Theorem~\ref{thm:P-CTC-state-distinguish} then proves
that P-CTCs can help distinguish an arbitrary set
of linearly independent states, but they are of no use for helping to
distinguish linearly dependent states. We also discussed how three different
ontological descriptions of a state (equivalent in standard linear quantum
mechanics) do not necessarily lead to the same consequences in the nonlinear
theory of postselected closed timelike curves. Finally, we provided explicit
P-CTC-assisted circuits that efficiently factor an integer, solve any decision
problem in the intersection of NP and co-NP, and probabilistically solve any
decision problem in NP (all using just one P-CTC\ qubit).

\textit{Acknowledgements}---We acknowledge useful conversations with Charles
H.~Bennett, Hilary Carteret, Patrick Hayden, Debbie Leung, and Graeme Smith.
We also acknowledge the anonymous referees for helpful comments.
TAB\ acknowledges the support of the U.S. National Science Foundation under
Grant No.~CCF-0448658. MMW\ acknowledges the support of the MDEIE (Qu\'{e}bec)
PSR-SIIRI international collaboration grant.

\end{document}